\documentclass[lettersize,journal]{IEEEtran}
\usepackage{amsmath,amsfonts}
\usepackage{array}
\usepackage{textcomp}
\usepackage{stfloats}
\usepackage{url}
\newtheorem{problem}{\textbf{Problem}}
\usepackage{verbatim}
\usepackage{graphicx}
\usepackage{color}
\newtheorem{lemma}{\textbf{Lemma}}  
\newtheorem{proof}{\textbf{Proof}} 
\usepackage{cite}
\usepackage{subfigure}
\usepackage{caption}
\usepackage[linesnumbered,ruled,vlined]{algorithm2e}
\begin{document}

\title{DDPS: Dynamic Differential Pricing-based Edge Offloading System with Energy Harvesting Devices}

\author{ Hai Xue, \IEEEmembership{Member IEEE}, Yun Xia, Neal N. Xiong, \IEEEmembership{Senior Member IEEE}, Di Zhang, \IEEEmembership{Senior Member, IEEE}, Songwen Pei, \IEEEmembership{Senior Member IEEE},
\thanks{This work was supported in part by the National Natural Science Foundation of China (NSFC) under grant 61975124, and the Joint Funds of NSFC under grant U22A2001. \emph{(Corresponding author: Songwen Pei.)}

Hai Xue, Yun Xia, and Songwen Pei are with the School of Optical-Electrical and Computer Engineering, University of Shanghai for Science and Technology, Shanghai 200093, China (e-mail: hxue@usst.edu.cn, xiadayun99@gmail.com, swpei@usst.edu.cn). 

Neal N. Xiong is with the Department of Computer Science and Mathematics, Sul Ross State University, Alpine, TX 79830 USA (e-mail: xiongnaixue@gmail.com).

Di Zhang is with the School of Electrical and Information Engineering, Zhengzhou University, the Henan International Joint Laboratory of Intelligent Health Information System, the National Telemedicine Center, and the National Engineering Laboratory for Internet Medical Systems and Applications, Zhengzhou 450001, China (e-mail: dr.di.zhang@ieee.org).
}} 

\maketitle

\begin{abstract}

Mobile edge computing (MEC) paves the way to alleviate the burden of energy and computation of mobile users (MUs) by offloading tasks to the network edge. To enhance the MEC server utilization by optimizing its resource allocation, a well-designed pricing strategy is indispensable. In this paper, we consider the edge offloading scenario with energy harvesting devices, and propose a dynamic differential pricing system (DDPS), which determines the price per unit time according to the usage of computing resources to improve the edge server utilization. Firstly, we propose an offloading decision algorithm to decide whether to conduct the offloading operation and how much data to be offloaded if conducted, the algorithm determines offloading operation by balancing the energy harvested with the energy consumed. Secondly, for the offloading case, we formulate the game between the MUs and the server as a Stackelberg game, and propose a differential pricing algorithm to determine the optimal computing resources required by MUs. Furthermore, the proposed algorithm also reallocates computing resources for delay-sensitive devices while server resources are surplus after the initial allocation, aiming to make full use of the server computing resources. Extensive simulations are conducted to demonstrate the effectiveness of the proposed DDPS scheme.

\end{abstract}


\section{Introduction}

\IEEEPARstart{I}{nternet} of things (IoT) is a transformative technology that empowers large-scale, resource-constrained devices with software programming, sensors, and network technology, which brings us closer to the realization of internet of everything. According to ITU-R WP5D, the ubiquitous and intelligent IoT will be the main focus of the forthcoming 6G research towards 2035. However, when mobile IoT devices are deployed in remote areas and resource-scarce environments (e.g., forests, deserts, and oceans), they often struggle to support applications demanding substantial computation resources and energy.  Although mobile IoT devices are becoming more and more powerful, their computing and battery power capabilities are still insufficient due to the limitation of physical size, which poses a great challenge to perform computation-intensive tasks on IoT devices \cite{HHJKS, ZLXQ, YZGST}.

Mobile edge computing (MEC) is expected as an effective solution to address the feeble computing capability and limited battery power challenges by offloading computational tasks from mobile IoT devices to MEC servers. MEC deploys servers with abundant computational resources at the network edge, minimizing the requirement of long-distance data transmission,  which results in low latency and prolongs operational longevity for mobile IoT devices. Nonetheless, the IoT device necessitates periodic recharging or battery replacement, particularly in remote areas which incurs substantial financial costs. Therefore, in this paper, we consider the edge offloading scenario with energy-harvesting devices to provide a sustainable power source for mobile IoT devices, for which not only mitigates the requirement for frequent manual battery replacements but also extends the overall lifetime of these mobile IoT devices.

Beyond the consideration of solar power, terminal devices also have to prioritize cost minimization. While terminal devices offload tasks to edge servers, they must carefully evaluate the trade-off between offloading operation and local execution costs to maximize their utility. To this end, questions such as whether to conduct offloading operation, how much data to be offloaded, and the optimal amount of computing resources to request from the edge server become paramount. Simultaneously, it is vital to foster an incentive mechanism to stimulate edge servers to willingly engage the task process for edge devices. Without adequate motivation, servers are disinclined to allocate resources to process the offloaded computation-intensive tasks. To this end, our focus shifts to the development of pricing strategies to incentivize server participation. That is to say, it is challenging to develop a trade-off pricing scheme between mobile users (MUs) and edge server while fully utilizing the server computing resources.

Most of the existing pricing schemes denoted the cost as an optimal constant per CPU cycle \cite{XWSZX, JSLY, CJZ},  and the unit price remains consistent with each unit of data length.  In other words,  the payments are directly proportional to the quantity of offloaded data with these existing schemes.  That is, if multiple users\footnote{ Unless stated otherwise, we use the terms MUs, terminal devices, and users, interchangeably.} offload the same amount of data and process identical applications, they would pay identical charges, irrespective of their computational demands on the server. To this end, MUs may be inclined to occupy substantial computational resources on the server. However, from the perspective of the server side, it is imperative to consider each user's utilization of computation, given the finite nature of server resources. Elastic pricing schemes that consider both computing resource usage and the volume of offloaded data alleviate the concerns of resource over-utilization. As a consequence, pricing strategies should encompass considerations of computing resource usage alongside offloaded data volumes.

Therefore, in this paper, we propose an effective dynamic differential pricing system with energy harvesting devices (DDPS). The proposed DDPS scheme considers each MU usage of server computing resources when determining the unit price. At the meantime, energy-harvesting devices are also utilized to calculate the optimal data offloading volume for MUs. For the game of terminal device and server, we employ the Stackelberg game to establish a cooperative decision-making process, which involves determining the computation resources required by the terminal devices and devising a corresponding pricing strategy for the server.

The contributions of this paper are summarized as follows.
\renewcommand\labelenumi{\theenumi)} 
\begin{enumerate}
    \item We firstly propose an efficient task offloading decision algorithm, which incorporates energy harvesting capability into the sensor devices to determine the volume of data offloaded based on the amount of harvested energy.
    
    \item  We propose a dynamic differential pricing algorithm taking into account each MU’s usage of server computing resources.  For the proposed pricing algorithm, the unit price is not fixed at a specific value, but a variable function of CPU usage of MU. Exactly speaking, MU should pay in proportion to the usage of server computing resources as well as the amount of offloaded data.

    \item  We establish a single-server multi-user Stackelberg game to determine the optimal computing resources required by MUs and the optimal server pricing. In addition,  we also propose a redistribution strategy for the situation where the server has surplus resources to maximize the server's utility.

    \item Extensive simulations are conducted to verify the effectiveness of the proposed DDPS scheme, the simulation results demonstrate that our proposed scheme enhances the utility of server, MUs ratio of service, and reduces the average delay of MUs.
    
   
\end{enumerate}

The rest of this paper is organized as follows.  In Section II, we present the related work. In Section III, the proposed system model is illustrated, including the model of energy consumption and delay. Section IV exhibits the problem formulation and solutions. In Section V, simulation results are presented and discussed. The conclusion is given in Section VI.

\section{Related work}

In the realm of MEC systems, extensive research has been conducted on resource allocation and pricing strategies to incentivize edge servers to serve users. We divide them into 3 categories from different perspectives as follows.

\subsection{Game pricing}
First of all, a common method is applying the Stackelberg game to address challenges related to resource provisioning and pricing between the MUs and server.
Tao $et$ $al.$ \cite{MKMH} utilized a Stackelberg game framework to establish the relationship between MEC server resources pricing and offloading data volume. The authors employed a differential evolution algorithm to seek the optimal pricing strategies.
Chen $et$ $al.$ \cite{YZBKK} decomposed the multi-resources allocation and pricing problem into sub-problems, constructing Stackelberg games for each sub-problem. The authors proposed an iterative algorithm to find equilibrium prices.
Li $et$ $al.$ \cite{YLYDY} introduced a novel multi-leader single-follower Stackelberg game, where each leader sets an optimal price based on the computation resources required by followers.
Qin $et$ $al.$ \cite{W} modeled the interaction between MEC servers and vehicles as a Stackelberg game, presenting a dynamic iterative algorithm to find Nash equilibrium for pricing determination.
Wang $et$ $al.$ \cite{MLPXKK} framed the interaction between UAV-MEC servers and MUs as a Stackelberg game, developing an arithmetic descent-based MRIG algorithm for computing resource price.
Mitsis $et$ $al.$ \cite{GEES} modeled user data offloading decisions as a non-cooperative game and employed semi-autonomous game methods or fully autonomous reinforcement learning to obtain optimal computing service price.
Li $et$ $al.$ \cite{QHTCY} formulated the resources management and pricing problem between MEC server and IoT devices as a bidirectional auction game, introducing the EWA algorithm to incorporate artificial intelligence for adaptive learning of optimal price.

\subsection{Auction pricing}
A large body of literature has used auction theory to study resource pricing in the context of MEC.
Hai $et$ $al.$ \cite{THP} addressed the resource allocation challenge between MUs and microclouds. The authors proposed a suitable auction mechanism to determine the price charged by microclouds to MUs.
Shen $et$ $al.$ \cite{ZJH} delved into edge cloud pricing issues, introducing an FPTAS auction mechanism to achieve socially optimal welfare.
Ng $et$ $al.$ \cite{JS} presented a full-payment auction mechanism to incentivize edge devices to actively participate in encoding computational tasks. For this mechanism, the bid from each edge device indicates its CPU capability, influencing the allocation of computational tasks.
Wang $et$ $al.$ \cite{QSJCL} firstly established the relationship between the resources offered by edge clouds and the charges to MUs in a non-competitive environment. Subsequently, in a competitive environment, the authors designed an online PMMRA auction mechanism, effectively determining the price paid by MUs for the resources provided by MEC servers.
Sun $et$ $al.$ \cite{WJYP} proposed two auction mechanisms, DAMB and BFDA, where MUs declare bids when requesting multi-task services, and edge servers collaboratively provide services to MUs.
Wu $et$ $al.$ \cite{BXYY} designed an auction mechanism for a MEC system comprising multiple MUs and a service provider. They presented a precise algorithm to maximize social welfare and a perturbation-based randomized allocation algorithm to achieve an approximate (1-$\alpha$) optimal social welfare, demonstrating the effectiveness of their auction mechanism.
Ma $et$ $al.$ \cite{LXXLYM} introduced a TCDA auction mechanism, providing distinct pricing strategies based on critical values and VCG mechanisms for MUs and edge servers. 
Wang $et$ $al.$ \cite{RCPYD} devised an auction pricing strategy where the highest bidder in each round could perform task offloading, confirming the efficacy of this strategy.
Su $et$ $al.$ \cite{YWYF} proposed a TCA auction mechanism and verified its effectiveness in addressing resource allocation and pricing challenges.

\begin{figure}[!t]
\captionsetup{singlelinecheck = false, justification=justified}
\centering
\includegraphics[width=3in]{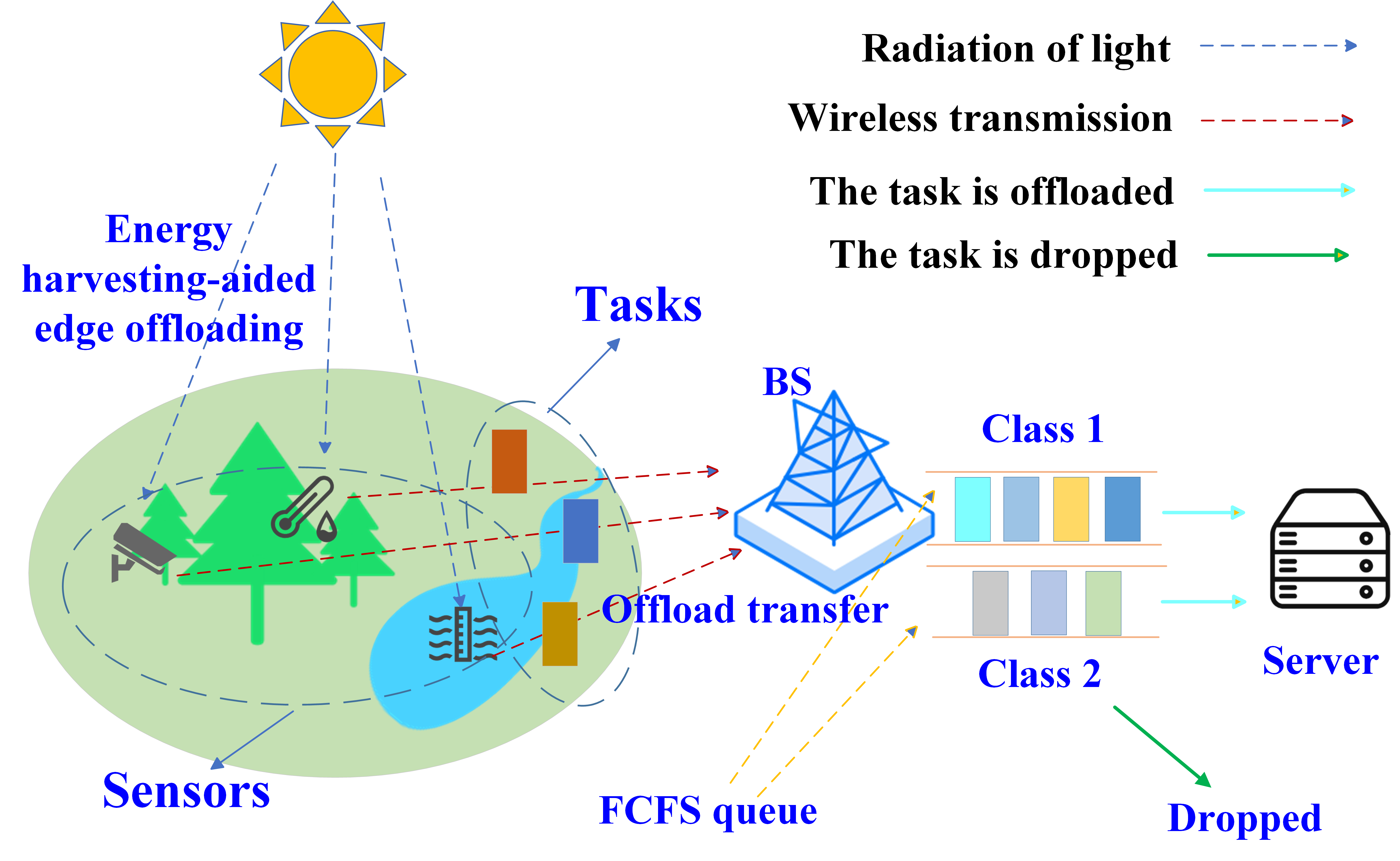}
\caption{System model.}
\label{fig1}
\end{figure}

\subsection{Dynamic Pricing Based on Physical Layer Parameters}
While existing pricing works often rely on abstract concave utility functions, some literature explores dynamic differential pricing based on MUs' physical layer parameters, contributing to a more nuanced understanding of resource pricing in MEC systems.
Han $et$ $al.$ \cite{DWY} utilized idle computation resources from parked vehicles in MEC systems, the authors dynamically adjusted price based on the current system state to minimize average costs.
Chang $et$ $al.$ \cite{ZLCYHY}  employed binary search to find optimal offloading delay to guarantee the processing latency, and dynamically adjusting price within the latency range.
Liu $et$ $al.$ \cite{MY} proposed unified and differential pricing strategies based on the extent of edge cloud awareness of network information.
Seo $et$ $al.$ \cite{HHJKS}, Liang $et$ $al.$ \cite{BRHYNA}, and Kim $et$ $al.$ \cite{SSMC} determined pricing dynamically based on MU utilization of server computing resources, with different approaches such as linear, quadratic, and profit-maximization formulations.

Although existing pricing schemes yielded excellent success, they still suffer from several issues. Most of these schemes\cite{XWSZX, JSLY, CJZ, MY, SSMC} define cost as a fixed optimal value per CPU cycle, resulting in a consistent unit price for each data length unit. That is, the payment amount becomes directly proportional to the data volume, disregarding MUs' computational demands on the server when multiple users offload and process equal amounts of data and applications. This may result in excessive utilization of computing resources on the server. Some papers \cite{LMZTQSQ, QL, WLPXKK} fail to consider both the quantity of offloaded data and computing resources usage, instead relying on an abstract concave utility function to uniformly price all MUs seeking server access. In other words, regardless of their specific requirements for offloaded data or computing resources used, all MUs are charged equally to maximize total social welfare. By doing so, this approach leads to reduced server utility and user satisfaction due to its inherent unfairness. From a server perspective with limited computing resources available, it is crucial to consider individual user utilization of these resources. Elastic pricing schemes help alleviate resource over-utilization problems by dynamically adjusting prices based on demand fluctuations.


\section{System model}

As depicted in Fig. \ref{fig1},  a single server and multiple MUs are considered in this paper, and we assume that all the tasks follow the FCFS rule in the queue.  MUs are charged for utilizing server computing resources, as long as MEC servers are available to handle the tasks in the queue. We denote that the amount of CPU resources utilized by a MU is decided by the maximum delay. Upon the task arrival, it first enters the first level queue (i.e., Class 1). if the server does not have enough computation resources to execute the current task under the delay constraint, the task is forwarded into the second level queue (i.e., Class 2) and waits for service, so that Class 1 generates a game between the subsequent task and the server \cite{LMZTQSQ}. At the beginning of the next time slot, the server assigns priority to serve tasks in Class 2, and then returns back to Class 1. For the task in Class 2, if the computing resources are still not enough, it continues to wait until the task exceeds the maximum delay and is dequeued. For the sake of convenience, Table I summarizes the important notations in this paper.

\begin{table}[!t]
\renewcommand{\arraystretch}{1.3}
\caption{SUMMARY OF  PRIMARY NOTATIONS}
\label{table_1}
\centering

\begin{tabular}{ c | >{\centering\arraybackslash}m{7cm}}
\hline

\textbf{Symbol} & \multicolumn{1}{c}{\textbf{Description}} \\
\hline
    $l_m$ &  The maximum amount of data offloading of the user\\
\hline
    $l_i$ &  The amount of data generated by the user\\
\hline
    $q_o$ &  The optimal data offloading amount of the user\\
\hline
$q_i$ & The amount of data offloaded by the $i$th user\\
\hline
$F_t$ & The remaining computing resources of the server\\
\hline
$F_i$ & The computing resource requested by $i^{th}$ user from server\\
\hline
    $F_{loc}^i$ &  The  local computing capability of MU $i$\\
\hline
    $h$ &  The number of cycles required to process a bit\\
\hline
    $A$ &  The area of the solar energy harvesting device\\
\hline
    $H$ &  The average amount of solar radiation\\
\hline
$\eta$ & The efficiency of the solar device\\
\hline
$k_a$ & The correction factor\\
\hline
$E_h$ & The energy harvested by a solar installation\\
\hline
$E_u$ & The amount of energy consumed by uploading task\\
\hline
$E_d$ & The amount of energy consumed by downloading  task\\
\hline
$k$ & The effective switched capacitance\\
\hline
$\mu$ & The service rate of the queue\\
\hline
$\lambda$ & The task arrival rate\\
\hline
$t_{req}^i$ & The delay requirement to be satisfied by the $i^{th}$ MU\\
\hline
$X$ & The pointer to the  Class 1 queue\\
\hline
$Y$ & The pointer to the  Class 2 queue \\
\hline
\end{tabular}
\end{table}

\subsection{MU energy consumption model}
 We assume that the total amount of data per task is $l_i$ bit, the amount of offloaded data is $q_i$ bit, and the server takes $h$ cycles to process a bit of data. Therefore, the energy consumption of local execution is expressed by the following formula \cite{RCPYDD}.
\begin{equation}
    E_{loc}^i=k \cdot (l_i-q_i)h \cdot F_{loc}^2.
    \label{eq_1}
\end{equation}
where $k$ denotes the effective switched capacitance based on the chip architecture, and $F_{loc}$ denotes the CPU capacity of the MU. We also assume that the achievable uplink and downlink transmission rates from the user to the server are $R_u$ and $R_d$, respectively.  Since the data size is different from the original size after execution if the amount of task $q_i$ is offloaded to the server, and we denote the ratio as $r$ \cite{YMXLJ},  then, the downlink data is $q_ir$. To sum up, the uplink transmission time $t_{up}$ and the downlink transmission time  $t_{down}$ are expressed as follows. 
\begin{equation}
    t_{up}=\frac{q_i}{R_u}. 
\end{equation}
\begin{equation}
    t_{down}=\frac{q_ir}{R_d}.
\end{equation}

Assume that the upload and download transmission power are $P_u$ and $P_d$, respectively. Based on this, the upload transmission energy consumption $E_u$ and the download transmission energy consumption $E_d$ are accordingly expressed as follows.

\begin{equation}
\begin{aligned}
    E_u&=P_u t_{up}\\
    &=P_u \frac{q_i}{R_u}.
    \label{eq_4}
\end{aligned}
\end{equation}
\begin{equation}
\begin{aligned}
    E_d&=P_d t_{down}\\
    &=P_d \frac{q_ir}{R_d}.
    \label{eq_5}
\end{aligned}
\end{equation}

Therefore,  the total energy consumption $E_{tot}$ of MU $i$ per time slot $t$ is obtained as follows.
\begin{equation}
    E_{tot}^i(t)=E_u^i(t) + E_{loc}^i(t) + E_d^i(t).
\end{equation}

In the meantime, the battery itself cannot supply power to the device for a long time since we assume all MUs are difficult to replace batteries in remote areas. However, energy harvesting devices are deployed as aforementioned, and the energy conversion $E_h$  is defined by the following equation.
\begin{equation}
    E_h^i(t)=AH\eta k_a.
    \label{eq_7}
\end{equation}
where $A$ is the area of the solar panel, $H$ is the average amount of solar radiation, $k_a$ $\in$ [0, 1] is the correction factor, and $\eta$ $\in$ (0,1) is the efficiency of the solar device. It should be noted that correction factors are added to represent the actual energy absorbed by energy devices due to objective reasons such as weather, shelter, etc. To sum up, the energy surplus $E_r$ per time slot $t$ of the device is expressed by the following formulas.
\begin{align}
    E_r^i(t)&=E_b^i(t)+E_h^i(t)-E_{tot}^i(t).\label{eight}\\
    \textbf{s.t.} &\quad E_r^i(t) \geq 0,\tag{\ref{eight}{a}}\\
    &\quad E_r^i(t) \leq E_b.\tag{\ref{eight}{b}}
\end{align}
where $E_b$ is the battery capacity of  MU $i$.

\subsection{Queue model}
Assume that the task arrival follows Poisson distribution and the probability function is expressed as the following formula.
\begin{equation}
    P(N=n)=\frac{e^{-\lambda}}{n!}\lambda^n.
\end{equation}
where  $\lambda$ is the task arrival rate and the $N$ is the number of tasks. The average waiting delay $t_w$ is calculated as follows.

\begin{equation}
    \begin{aligned}
    t_w&=\frac{1/\mu}{1-\lambda/\mu}-\frac{1}{\mu}\\
    &=\frac{\lambda/\mu^2}{1-\lambda/\mu}.
\end{aligned}
\end{equation}
where $\mu$ is the service rate of the queue which represents the number of tasks in the queue processed per unit time. For the proposed scheme, it is assumed that the task is completed within the maximum delay. So,  the average completion time $t_{ave}$ is approximately denoted as the mean value of the maximum delay of all MUs. $t_{ave}$ is derived as follows.
\begin{align}
    \mu &= \frac{\max N}{t_{ave}}.\label{eleven} \\
     \textbf{s.t.}& \quad F_t \leq \sum_{i=1}^NF_i. \tag{\ref{eleven}{a}}
\end{align}

\begin{equation}
    t_{ave}=\frac{\sum_{i=1}^{\max N}t_{req}^i}{\max N}.
\end{equation}
where $t_{req}^i$ is the delay requirement to be satisfied by the  $i^{th}$ MU, and $F_i$ is the server CPU capacity utilized by the $i^{th}$ MU,  $F_t$ denotes the available  CPU capacity of the server.  $F_i$ and $F_t$ are illustrated in subsection $C$ detailedly.

\subsection{Pricing model for the server}
Denote $F_i$ as the CPU capacity utilization by the $i^{th}$ MU, and the processing time $t_p^i$ at the server is shown as follows.
\begin{equation}
    t_p^i=\frac{hq_i}{F_i}.
    \label{eq_13}
\end{equation}

Here, a pricing bidding function is proposed which considers the CPU utilization to set the price per unit time. The principle of defining the bidding function is that the function must be an increasing function \cite{HHJKS}. The bidding feature should become more expensive as the user consumes more CPU capacity, and the server should be set at a higher price to ease the pressure on the server as it has fewer available resources. Let $F_t$ denote the available server CPU capacity, the unit price is defined as the product of the pressure and the resources consumed, which is illustrated as follows.
\begin{equation}  
    w_i=\frac{F_i}{F_t}f(F_i).
\end{equation}
where $w$ is the unit price. Finally, the total payment is the product of the price per unit time and the processing time, which is calculated as follows.
\begin{equation}
    \begin{aligned}
         W_i &= wt_p^i \\
           &= \frac{hq_i}{F_t}f(F_i)  .
    \end{aligned}
\end{equation}
where $W_i$ is the payment of the $i^{th}$ MU.  In addition,  a discount strategy exists when the MU uses more server resources, we adopt the logarithmic function\footnote{The $lg(\cdot)$ function we priced is based on a logarithm with a base of 10, this is because computational resources are usually represented in scientific notation, and which is convenient for calculation.} as the price bidding function \cite{YZBKK, HHJKS}. Then, the payment $W_i$ is expressed as 
\begin{equation}
    \begin{aligned}
     W_i = \frac{hq_i}{F_t} \lg (F_i +d).
     \label{16}
    \end{aligned}
\end{equation}

Here, $d$ should not be smaller than 1  since the payoff is always positive. It is important to note that $W_i$ should not be too large as well as $w_i$, otherwise, the user pays too much and abandons the task processing, which increases the packet loss rate as well as the corresponding penalty function increases. The penalty function $P$ is defined hereinafter.
\begin{lemma}
    The payment $W_i$ is an increasing function concerning the offload data size $q_i$ and the CPU capacity of the server utilized by MUs $F_i$, respectively.
\end{lemma}
    
\begin{proof}
$h$ is the required number of CPU cycles for one-bit data, and $F_t$ is the total CPU capacity of the server, they should be non-negative numbers by definition. The partial derivatives of payment concerning $q_i$ and $F_i$ are as follows:

\begin{equation}
    \frac{\partial W_i}{\partial q_i}=\frac{h}{F_t}\lg(F_i+d)\geq 0.
\end{equation}
\begin{equation}
    \frac{\partial W_i}{\partial F_i}=\frac{hq_i}{F_t}\frac{1}{(F_i+d)\ln10} \geq 0.
    \label{18}
\end{equation}
\end{proof}

For either $q_i$ or $F_i$, their first partial derivatives are greater than zero, so $W_i$ is an increasing function.

\begin{figure}[tp]
\captionsetup{singlelinecheck = false, justification=justified}
\centering
\includegraphics[width=3in]{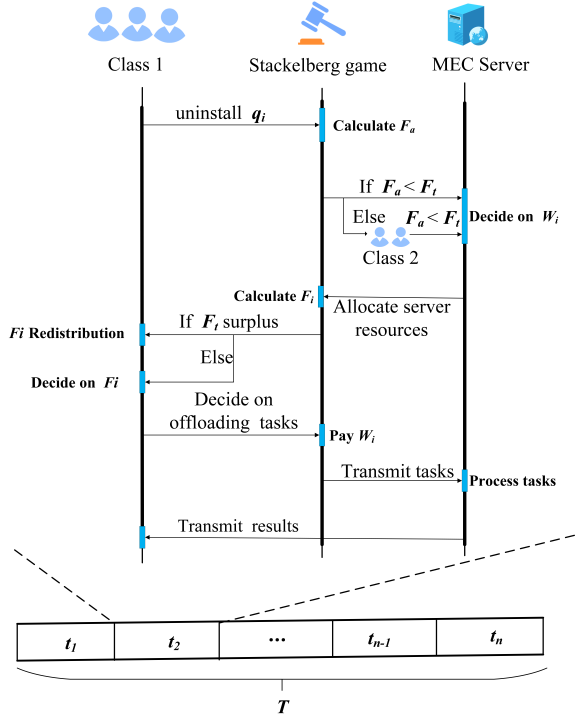}
\caption{The procedure of Stackelberg game.}
\label{fig2}
\end{figure}

\section{The proposed DDPS}
In this section, we model the game between the server and MUs as a Stackelberg game with a single leader and multiple followers, where the server acts as a leader and multiple MUs act as followers. As depicted in Fig. \ref{fig2}, which exhibits the procedure of the established Stackelberg game between MUs and MEC server.  When a MU intends to offload a task, it firstly provides task-related information, including its  CPU capacity ($F_{loc}$), the required server CPU capacity ($F_i$), the minimum latency requirement ($t_{req}$), and the data size to be processed ($l_i$). Subsequently, the server evaluates the information to decide whether it can accommodate the task. If acceptable, the server then informs the user about the remaining server resources and quotes an appropriate service price to the user. According to the game rule we defined, aiming to maximize server utility function, when all MUs enter the server queue within a time slot and there are remaining server resources, the surplus computing resources are allocated based on the ratio of each MU's offloaded data amount to the total server data amount, and then quote the higher price accordingly to increase user average expenditures. 

Subsequently,  a cost function for the user and a utility function for the server is formulated, respectively. Each user selfishly minimizes its cost function and the server selfishly maximizes its utility function, separately.

\subsection{The utility function of MU and the optimal policy}
The utility function of MU  is obtained by subtracting the payment to the edge server from the profit of energy saving by offloading tasks.  So, the utility function $U_{user}^i$ of the MU is expressed as follows.

\begin{equation}
    U_{user}^i=\max E_r^i + \min W_i.
    \label{eq_19}
\end{equation}

Subsequently, the execution latency is modeled, for which a partial offloading model is considered \cite{LSYZ}. The MU offloads part of the data and the $l_i$ is processed both on the server and locally. To this end, the execution delay is expressed as the higher value between the time taken for offloading $t_{off}$ and the local processing time $t_{loc}$, which is expressed as follows.
\begin{equation}
     t_e=\max (t_{loc}, t_{off}).
     \label{eq_20}
\end{equation} 

The time $t_{loc}$ required for all tasks to be executed locally is shown as follows.
\begin{equation}
    t_{loc}=\frac{(l_i-q_i)h}{F_{loc}}.
\end{equation}

The time consumed for processing the offloaded data is set as the sum of waiting delay $t_w$, uplink data transmission time $t_{up}$, processing time on the server $t_p$, and downlink data transmission time $t_{down}$.
\begin{equation}
    t_{off}=t_{up}+t_p+t_{down}+t_w.
\end{equation}

Since the downlink transmission data is less than the uplink transmission data $q_i$, and the time consumed for offloading is expressed as follows.
\begin{equation}
    \begin{aligned}
    t_{off}&=\frac{q_i}{R_u} + \frac{rq_i}{R_d} +t_p+t_w\\
            &=(\frac{1}{R_u}+\frac{r}{R_d} +\frac{h}{F_i})q_i+\frac{\lambda/\mu^2}{1-\lambda/\mu}.
    \label{eq_23}
\end{aligned}
\end{equation} 

When the offloading delay exceeds the delay tolerance constraint of the task, then, the task is dropped and the penalty function $P$ of packet loss is given by the following equation \cite{DWY}. 
\begin{equation}
    P=\gamma\sum_{i=0}^{M}W_i.
    \label{eq_24}
\end{equation}
where $M$ is the number of dropped packets and $\gamma$ $\in$ (0, 1) is a weight coefficient. \textbf{It should be noted that since the discarded tasks do not play a game with the server, the optimal $F_i$ cannot be determined. Here, we calculate the $F_i$ based on the maximum delay tolerance as the pricing directly.}

\begin{algorithm}[!t]
    \SetAlgoLined
    \caption{Offloading decision\protect\footnotemark }\label{alg:alg1}
     \KwIn{$t_{req}$}
     \KwOut{$X$}
     \While{X!=null}{
     Enter the $X.l$ \;
    Get the $E_h$ according to Eq. (\ref{eq_7})\;
    According to \textbf{Lemma2}, \textbf{Lemma3} and \textbf{Lemma4}, get $l_c$, $q_o^{opt}$, $l_m$.\;
    \eIf{$0<X.l<l_c$}{
        compute locally\;
    }{
        \eIf{$l_c \leq X.l < q_o^{opt}$}{
        offload partially\;
        calculate the X.F according to the $t_{req}$\;
        calculate the X.t according to the Eq. (\ref{eq_20});
        }{
            \eIf{$q_o^opt \leq X.l \leq l_m$}{
            calculate the X.F according to the $t_{req}$\;
            calculate the X.t according to the Eq. (\ref{eq_23});
            }{
            drop;
            }
        }
    }
    X=X.next;
    }
\end{algorithm}
\footnotetext{When designing the algorithm, for the variables involved in the table, we omitted the subscripts and used Pointers to represent the $i^{th}$ user.}

Aiming to obtain the optimal strategy for the MU, a budgeted energy offloading strategy is proposed based on the pricing strategy of the server. $U_{server}^i$ consists of two functions, one is to maximize the residual energy $E_r^i$, and the other is to minimize the expenditure $W_i$. For the former, according to the amount of harvested energy $E_h^i$, the task is executed locally or offloaded, and then the optimal amount of offloading data $q_i^{opt}$ is obtained, which is referred to Lemma \ref{lemma2}, Lemma \ref{lemma3}, Lemma \ref{lemma4}. As a consequence, Eq. (\ref{eq_19}) is converted into the following optimization problem.

\begin{problem}
    Optimal Offloading Decision
\end{problem}

\begin{equation}
    \begin{aligned}
    \label{eq_25}
    U_{user}^i&=\min W_i\\
    &= \min \frac{hq_i}{F_t} \lg (F_i +d).
\end{aligned}
\end{equation}
\begin{equation}
    \textbf{s.t.} \quad 0\leq q_i\leq l,\tag{\ref{eq_25}{a}}
\end{equation}
\begin{equation}
        \quad \quad \quad 0\leq F_i\leq F_t,\tag{\ref{eq_25}{b}}
\end{equation}
\begin{equation}
    \quad \quad \quad 0\leq t_e \leq t_{req}^i.\tag{\ref{eq_25}{c}}
\end{equation}

According to Eq. (\ref{eq_25}), it is an increasing function concerning $q_i$,  which implies offloading data as little as possible. On the contrary, $W$ is a decreasing function concerning $F_i$, which indicates the selection of the minimum computing resources. However, fewer computation resources required lead to longer computing time. Therefore, it is necessary to complete the task within the maximum delay tolerance. To this end, MUs prefer to offload the smallest amount of data and less computing resources to be optimal for them.  However, since each MU tends to use fewer server resources, it takes a longer time to complete the task. In addition, it causes slow task turnover and less profit for the server, so the server does not allow each MU to achieve the minimum computing resources to maximize its profit. 

To solve the optimal offloading decision problem (i.e., Problem 1), we proposed Algorithm \ref{alg:alg1} to determine the eligibility of offloading conditions for the user. According to the harvested energy, and considering its maximum task delay and cost factors, the MU decides whether to execute locally or partially offload, and its time complexity is $O(n)$.


\begin{lemma}\label{lemma2}
    We define $l_c$ as the critical amount of data whether to offload or not.
\end{lemma}
\begin{proof}
    When the harvested energy is just enough to process all the data locally, we define the data volume of the task as $l_c$. At this point, the following formula is obtained.
    \begin{equation}
        E_h=E_{loc},
    \end{equation}
    \begin{equation}
        AH\eta k_a=k \cdot l_ch \cdot F_{loc}^2,
    \end{equation}
    \begin{equation}
        l_c=\frac{AH\eta k_a}{k \cdot h \cdot F_{loc}^2}.
    \end{equation}
\end{proof}

In the meantime, we set $y$ as a binary variable, when $y$ equals 1, it means offloading. Otherwise, it means local processing. Then, we can obtain the following expression.
   
    \begin{equation}
        y= \left\{
        \begin{array}{ll}
        \text{1}, \quad \text{$l_i$ $\leq$ $l_c$}; \\
        \text{0}, \quad \text{$l_i$ $\geq$ $l_c$}. \\
        \end{array}
        \right.
   \end{equation}

\begin{lemma}\label{lemma3}
    When all the energy collected by the MU is used to offload the data, the amount of data offloaded by the MU reaches the maximum $l_m$ at this time.
\end{lemma}
\begin{proof}
     The maximum amount of data is processed when all the energy is used for offloading \cite{WGS, JZXRZT, MWXYLL}. At this point, we obtain the following expression.
    \begin{equation}
        E_h=E_u+E_d,
    \end{equation}
    \begin{equation}
        AH\eta k_a =P_st_{up} +P_st_{down},
    \end{equation}
    \begin{equation}
        AH\eta k_a =P_u\frac{l_m}{R_u} +P_d\frac{rl_m}{R_d},
    \end{equation}
    \begin{equation}
        l_m=\frac{AH\eta k_a}{B}.
    \end{equation}
    where $B=\frac{P_u}{R_u}+\frac{rP_d}{R_d}$.
\end{proof}

\begin{lemma}\label{lemma4}
    We define the balanced data $q_i^{opt}$  as the optimal amount of offloaded data when the harvested energy is just enough to meet the energy expenditure with the offloading case.
\end{lemma}
\begin{proof}
   When $l$ $\geq$ $l_c$, the energy consumed by the MU to offload the task includes the energy consumed by executing the task locally $E_{loc}$, the energy consumed by uploading the task $E_u$, and the energy consumed by downloading the result $E_d$.  Combining  Eq. (\ref{eq_1}), Eq. (\ref{eq_4}), Eq. (\ref{eq_5}) and Eq. (\ref{eq_7}),  we can get the following equation.
   \begin{align}
        E_h& \geq E_{loc}+E_u+E_d,\\
       AH\eta k_a& \geq k \cdot (l-yq_i^{opt})h \cdot F_{loc}^2 +y(P_ut_{up}+P_dt_{down}),\\
       l& \leq yq_o+\frac{AH\eta k_a-y(P_ut_{up}+P_dt_{down})}{k \cdot h \cdot F_{loc}^2},\\
       l& \leq yq_o+\frac{AH\eta k_a-yq_i^{opt}(\frac{P_u}{R_u}+\frac{P_dr}{R_d})}{k \cdot h \cdot F_{loc}^2}.\label{thirty-seven}\\
       &\textbf{s.t.} \quad l-l_c \leq q_o^{opt} \leq l_m.\tag{\ref{thirty-seven}{a}}
   \end{align}
\end{proof}

\subsection{The optimal price of the edge server}
The utility of the edge server, $U_{server}$, is equal to the payment of the user minus the penalty function $P$ to discard tasks, which is expressed as the following equations.
\begin{equation}
    \begin{aligned}
        U_{server}&=\sum_{i=1}^{M}W_i-P\\
        &=\frac{h}{F_t} \lg (\prod_{i=1}^{M} (F_i +d)^{q_i})-\gamma\sum_{i=0}^{M}W_i. 
    \label{eq_38}
    \end{aligned}
\end{equation}

\begin{equation}
    \textbf{s.t.}\quad \sum_{i=1}^MF_i \leq F_t.\tag{\ref{eq_38}{a}}
\end{equation}

Appropriate bidding function should be made by the server to guarantee its utility as non-negative. Here, the bidding function is treated as a logarithmic function of $f(x)=\lg(F_i+d)$ \cite{YZBKK, HHJKS}.  The input is regarded as the usage rate of the MU at the server.  Since the bid function must be larger than 0, and the bid function should be 0 when $F_i$ = 0, that is, when the MU is not using the server resources, $d=1$; $F_i$  is the optimal offloading decision of the MU. However,  when $F_i$ is large, the service period profit is correspondingly large, but according to the game rule, MU is not willing to request a large $F_i$, and when $F_i$ is large, the number of users that can be served will be less. That is, as the penalty function increases the server utility becomes low. Therefore, within an acceptable range of $F_i$  to the user, the utilization of server resources is maximized, and server profit is maximized simultaneously. In this case, the optimization problem to maximize the $U_{server}$ based on the optimal policy of the MU is expressed as follows.
\begin{problem}
    Decision on price bidding function
\end{problem}
\begin{equation}
\begin{aligned}
    \label{opt_3}
     U_{server} &= \max(\sum_{i=1}^{M}W_i-P).
\end{aligned}
\end{equation}

\begin{equation}
    \textbf{s.t.}\quad \sum_{i=1}^MF_i = F_t.\tag{\ref{opt_3}{a}}
\end{equation}

To fully utilize the server resources and improve the revenue of the server, we add the restriction that each MU must increase $F_i$ from the original to fully use the server resources if idle resources exist. MU takes up the extra resources of the server according to the proportion of their offloaded data. As the MUs occupy more server resources,  they have to pay more, which implies the revenue of the server is improved. The justification for such a decision by referring to Appendix \ref{price}.
\begin{algorithm}[!t]
    \SetAlgoLined
    \caption{Differential pricing}\label{alg:alg2}
    \KwIn{ $R_u$, $R_d$, $F_{loc}$, $X$, $Y$}
    \KwOut{$U_{server}$, $U_{user}$}
    
    Check the $Y$\label{step1}\;

    \eIf{Y == null}{
        Check the $X$\;
        \If{X==null \textbf{and} $X_s$ == null  \textbf{or} $F_t \leq \epsilon$ \footnotemark}{
        the server didn't provide service;
        }
        \eIf{X==null \textbf{and} $X_s$ != null}{
         \eIf{$F_t>0$}{
            \For{X.u;X!=null;X=X.next}{
                $l_{tot}=X.l$;
            }
            $X_s.F=X_s.F+F_t*(\frac{X_s.l}{l_{tot}})$
            calculate the actual time $t_{fact}$ according to Eq. (\ref{eq_13});
         }{
         calculate the $U_{server}$ according to Eq.  (\ref{eq_38})\;
         calculate the $U_{user}$ according to Eq.  (\ref{eq_25});
         }
        }{
            \eIf{$F_t>X.F$}{
            X.u is served;\;
            $F_t = F_t - X.F$\;
            }{
            \eIf{$t_{req}>X.t$}{
            Y=X\;
            Y=Y.next;
            }{
            X.u is dropped\;
            calculate the P according to Eq. (\ref{eq_24});
            }
            X=X.next\;
            
            }
            
        }

    }
    {
        \While{$F_t \geq Y.F$}{
            Y.u is served\;
            $F_t = F_t - Y.F$\;
            Y = Y.next\;
        }
    }
\end{algorithm}    

\footnotetext[5]{$\epsilon$  is a critical value, when $F_t \leq \epsilon$, the server will not continue to access the next MU and serve the MUs already at the server.}
The joint optimization Problem 1 and Problem 2 are solved by the proposed Algorithm \ref{alg:alg2}, which determines the eligibility of a MU to be served by the server after offloading and calculates the utility for both the server and the user. At the beginning of each time slot, the server first iterates the tasks in  Class 2. If the server meets the computing resources requirement by the task in  Class 2, it processes the task. At this time, the server gives an optimal initial payment according to the amount of data offloaded and the required computing resources, otherwise, the task continues to wait. After iterating the tasks in Class 2, it continues to iterate the tasks in Class 1. If the computing resources of the server cannot meet the computing resources requirement by the task in Class 1, the task is forwarded into Class 2  and waits. Since the server iterates over Class 1 and Class 2, the entire process takes $O(2n)$ time. For each MU already in the server, if the server has any remaining computing resources, we adopt a redistribution strategy, it needs to iterate through the array $X_s$ again,\footnote{$X_s$ represents the array of MUs to be served.} reallocates computing resources and calculates the execution delay, and the time complexity is $ O(n)$. Therefore, the entire time complexity is $ O(3n)$.

\section{Performance evaluation}


To validate the effectiveness of the proposed DDPS scheme, we compare it with four existing schemes: 1) the uniform pricing that the server picks an optimal price from a set of prices and charges all MUs by the price \cite{MY}; 2) the differentiated pricing that the unit price set by the server is inversely proportional to the computing resources requested by the user \cite{MY}; 3) the linear pricing that the unit price set by the server is linear with the computing resources requested by the user \cite{HHJKS}; 4) the nonlinear pricing that the unit price set by the server is a quadratic function of the computing resources requested by the MU \cite{BRHYNA}.
\begin{itemize}  
\item \textbf{Uniform pricing \cite{MY}}: The edge cloud uniformly sets and broadcasts a price $\mu$ to all users, denoted as $\mu = \mu_1 = ... = \mu_K$. Each user determines its optimal offloading strategy by solving the corresponding optimization problem for the given uniform price $\mu$. Meanwhile, the edge cloud determines its optimal price $\mu^*$ based on each MU's offloading decision. Subsequently, the edge cloud sequentially announces prices $\{1/F_{loc}^i\}^k_{k=1}$ to MUs in descending order until the computing power constraint is satisfied, thereby concluding the price negotiation process.

\item \textbf{Differentiated pricing \cite{MY}}: The equation $\mu^*_k = 1/F_{loc}^i$ implies that the optimal price for user $k$, denoted as $\mu^*$, is equal to the reciprocal of its  $F_{loc}^i$. User $k$  chooses  $m_k$ bits of data to offload if its  $F_{loc}^i$ is less than or equal to the threshold $1/\mu^*_k$; Otherwise, it opts for local computation. Consequently, the optimal price $\mu^*_k$ for user $k$ is dependent on its  $F_{loc}^i$, with a higher value assigned when the $F_{loc}^i$ is lower.

\item \textbf{Linear pricing \cite{HHJKS}}: The price function is defined as a linear function, represented by $f(x) = ax + b$, where $x$ represents the proportion of computing resources utilized by users on the server. The determination of values for $a$ and $b$ constitutes the decision variables for the server. By employing differentiated pricing, the server can ascertain an appropriate unit price based on the user's utilization of computing resources, thereby maximizing its utility.

\item \textbf{Nonlinear pricing \cite{BRHYNA}}: The price function  is defined as $P(x) = ax^2 + bx$, where $a$ and $b$ are non-negative parameters, and $x$ represents the computation capability used by the MU. The authors employ this quadratic function to approximate the price function $P(\cdot)$ while satisfying three properties: 1) when no computing service is provided, $P(0)$ equals zero; 2) $P(\cdot)$ exhibits monotonicity; 3) $P(\cdot)$ demonstrates convexity.

\end{itemize}

In summary, the two schemes of uniform pricing and differential pricing are solely dependent on the local computing power, whereas linear pricing and nonlinear pricing are contingent upon the server computational capabilities used by MUs.

\begin{figure}[htp]
\captionsetup{singlelinecheck = false, justification=justified}
\centering
\includegraphics[width=3in]{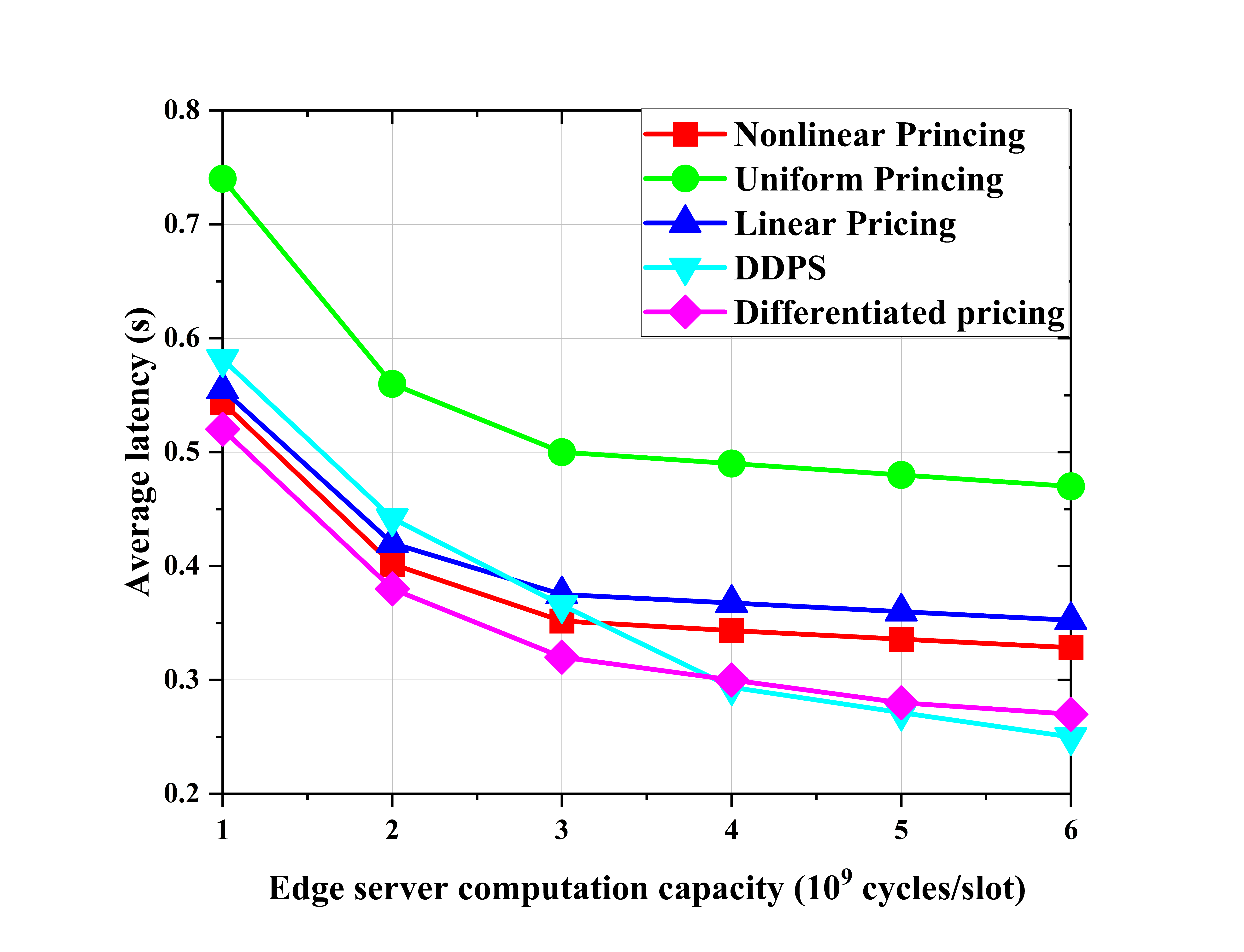}
\caption{Effect of different $\lambda$ on average latency.}
\label{fig3}
\end{figure}

\begin{figure}[htp]
\captionsetup{singlelinecheck = false, justification=justified}
\centering
\includegraphics[width=3in]{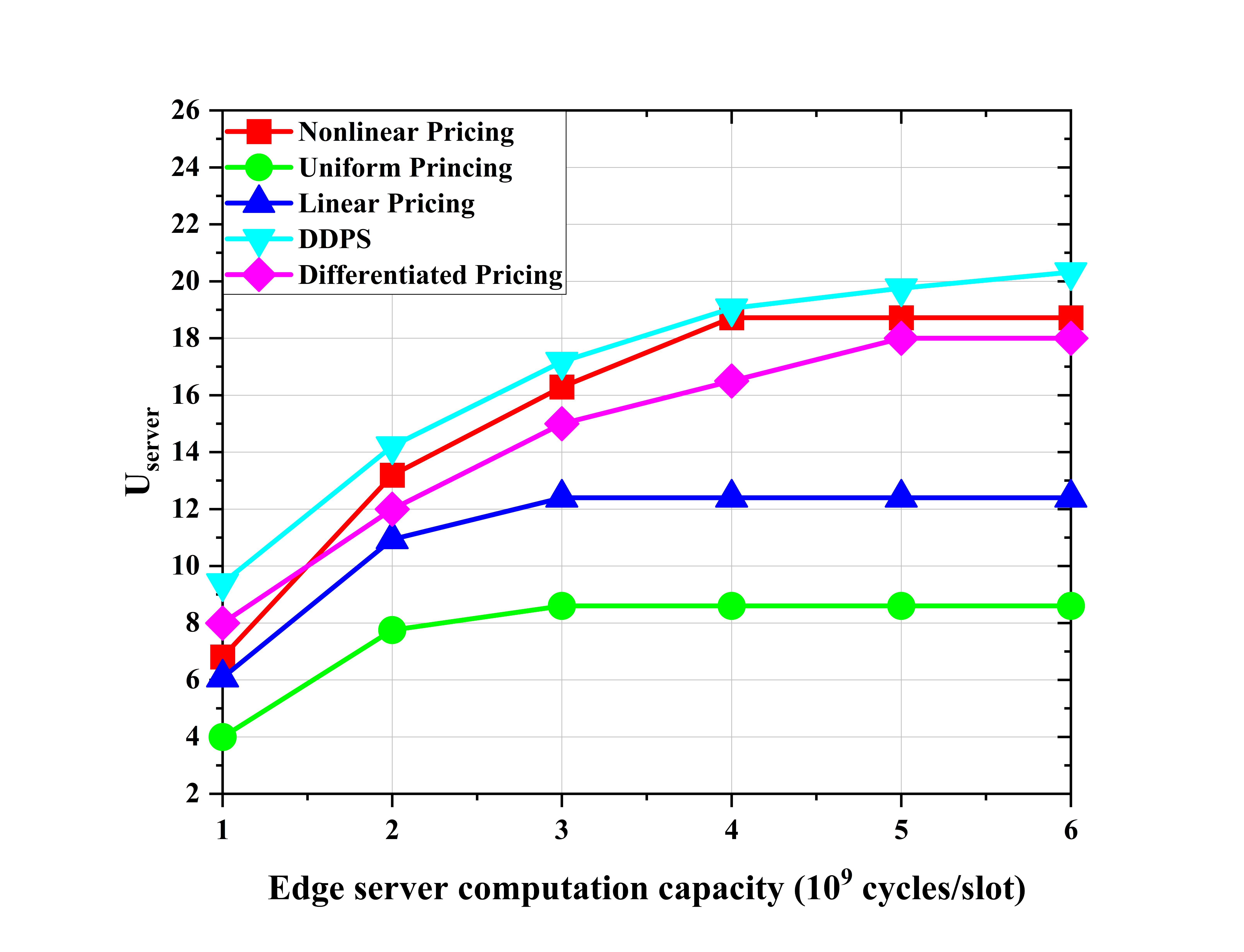}
\caption{Effect of different $\lambda$ on $U_{server}$.}
\label{fig4}
\end{figure}

\subsection{Simulation Setup}
For the simulation setup, we assume that   $\lambda$ of MUs are 0.1, 0.2, 0.3 \cite{HHJKS}. To specify the services provided by MEC servers, we consider face recognition applications that can be executed in the edge cloud.  For the parameters of delay, we rely on some existing works to get the representative values.  Specifically, the transmission delay has a very large variance, and the mean value is assumed to be 100 to 300 ms \cite{MWDPLLJ}.  The execution delay is set as 500 ms \cite{NAKTM}. The local CPU frequency $F_k$ for each MU $i$ is uniformly selected from the set \{0.1, 0.2, ..., 1\} GHz,  the data size for MU $i$ are uniformly distributed with  $l_i$ $\in$ [100, 500] KB. Unless otherwise noted, the remaining parameters are uniformly defined in Table \ref{tabel_2} \cite{MY}.

\begin{table}[!ht]
\begin{center}
\caption{default parameter settings}
\label{tabel_2}
\begin{tabular}{|c|c|c|c|c|c|}  
\hline  
\textbf{Parameter}  & $P_u$  & $P_d$ & $F_t$ & $r$ &  $\lambda$ \\ \hline  
\textbf{Value} & 0.1W & 1W & $6 \times 10^9$ cycles/slot & 0.2 & 0.3 \\ \hline 

\end{tabular}
\end{center}
\end{table}

\subsection{Impact of edge server computation capacity on average execution latency}

Fig. \ref{fig3}  exhibits the impact of the edge server computation capacity on average execution latency. The computation capacity is progressively augmented in increments of $1 \times 10^9$ cycles/slot, up to a maximum of $6 \times 10^9$ cycles/slot. 

As shown in Fig. \ref{fig3}, it is discernible that as the computational capacity of the edge server enhances, there is a concomitant diminution of the average latency across all pricing schemes. When the computing capacity is low, DDPS exhibits a relatively high average delay compared with the other schemes except for uniform pricing. The reason is that according to our Algorithm \ref{alg:alg1}, DDPS  selects the smallest $F_i$ under the maximum delay to reach the optimal offloading decision, which leads to a long average delay.  With the increase of server computing resources, Algorithm \ref{alg:alg2} initially allocates computing resources to MUs, and if any service computing resources remain, the server reallocates resources to delay-sensitive MUs, which accelerates the process of task execution.


In addition, from the perspective of pricing, in the case of low computing resources of the server, which means each MU  obtains low computing resources. Compared with linear pricing and nonlinear pricing schemes, the $log(\cdot)$ function of DDPS pricing is higher, which also suppresses MUs from requesting more computing resources. Similarly, with the increase of computing resources, the $log(\cdot)$ function of DDPS is priced relatively lower, which encourages MUs to request more computing resources to reduce the latency accordingly. For the uniform pricing scheme, no matter how many computing resources are used by MUs, the price is uniform for all MUs, which leads to unfairness to a part of MUs that just need little computing resources. It reveals that when the MU executes tasks locally more, resulting in the highest latency, which accords with the original intention of MEC.

\subsection{Impact of edge server computation on $U_{server}$}

Fig. \ref{fig4}  exhibits the impact of the edge server computation capacity on $U_{server}$. The computation capacity is varied as utilized in Fig. \ref{fig3}.

As shown in Fig. \ref{fig4}, with the increase of server computing resources, the revenue of all pricing schemes increases, and the DDPS scheme has the highest revenue. The utility function ($U_{server}$) of the edge server is used as a substitution for revenue to verify the effectiveness of the proposed scheme, and we run the simulations for 1 minute. The reasons are explained as twofold. On the one hand, DDPS allows each MU to select the least server computing resources under the maximum latency, which means that the server can serve more MUs under the same condition. On the other hand, we implement a resource redistribution strategy, so that the server resources are fully utilized. Existing pricing schemes (e.g., Uniform pricing and Differentiated pricing \cite{MY}) mainly focus on optimizing the cost per CPU cycle, consequently, when the amount of data offloaded requires the same number of cycles, each MU dedicates to occupying more server computing resources to reduce execution time, which leads to sub-optimal utilization of server computing resources. Therefore, the proposed DDPS scheme gains more revenue than other existing schemes.

\begin{figure}[tp]
\captionsetup{singlelinecheck = false, justification=justified}
\centering
\includegraphics[width=3in]{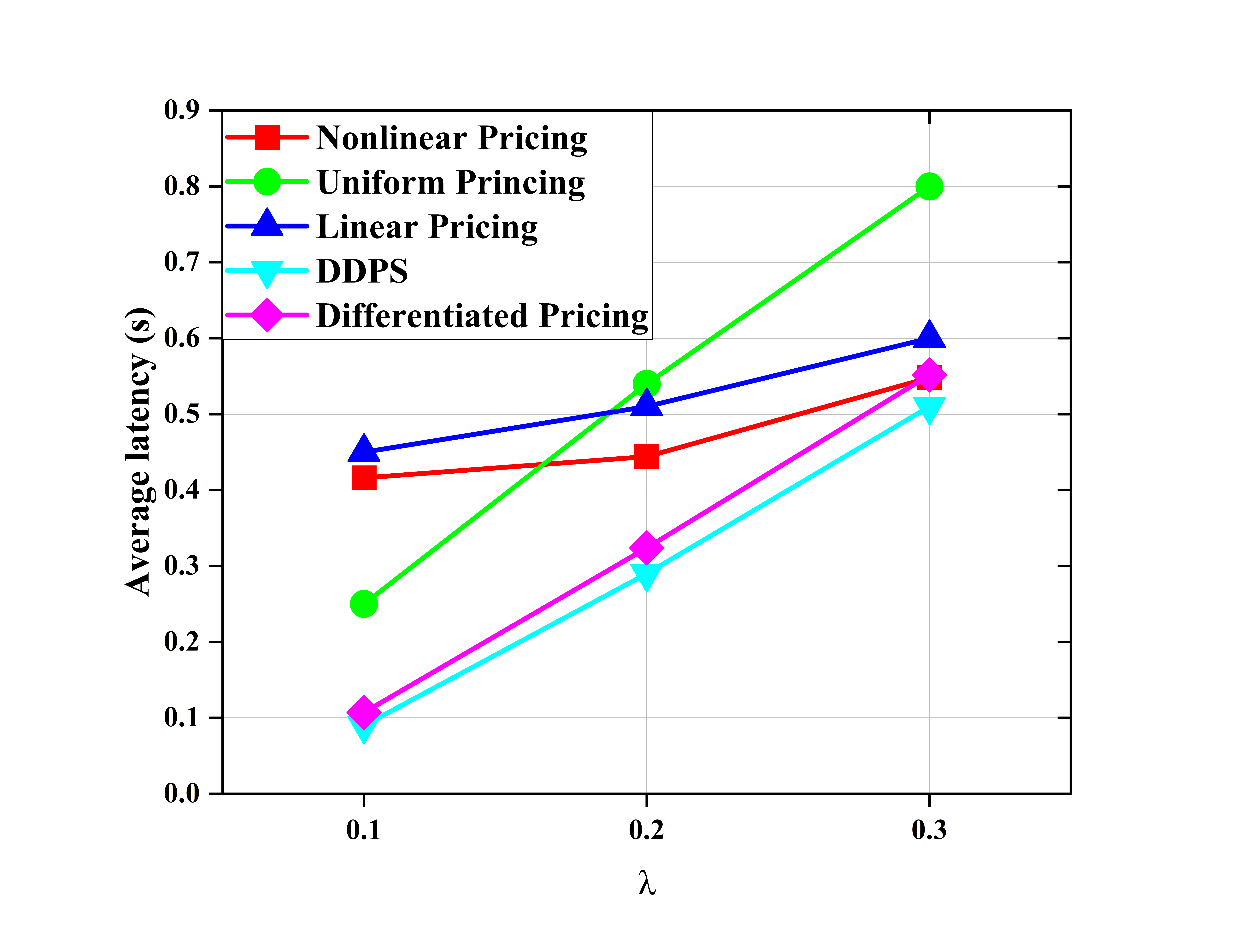}
\caption{Effect of different $\lambda$  on average latency.}
\label{fig5}
\end{figure}

\subsection{Impact of  $\lambda$ on average execution latency}

Fig. \ref{fig5} exhibits the impact of  $\lambda$ on the average execution latency. Here, $\lambda$ are set as 0.1, 0.2, and 0.3, respectively. As observed from Fig. \ref{fig5}, an uptick with $\lambda$ leads to a concomitant increase in the average latency across all pricing schemes, and the DDPS outperforms the other schemes in terms of average execution latency. This is because the redistribution strategy of DDPS makes full use of server resources, that is, each MU can get more computing resources to quickly complete the task execution. In addition, compared with linear pricing and nonlinear pricing, our pricing scheme is also conducive to MU offloading due to the lower price, thereby reducing the delay. Compared with the uniform pricing scheme, the superiority of our scheme is proved by Appendix  \ref{latency}.




\subsection{Impact of $\lambda$ on $U_{server}$}
Fig. \ref{fig6} exhibits the impact of  $\lambda$ on $U_{server}$. $\lambda$ are set as utilized in Fig. \ref{fig5}. The results presented in Fig. \ref{fig6} show that DDPS stably generates the highest revenue with the increase of $\lambda$. Compared with existing schemes, in addition to the high income caused by the redistribution strategy, DDPS also has the advantage of attracting more MUs at a lower price compared with linear and nonlinear pricing.  In addition, DDPS allocates the least amount of computing resources to each MU at the beginning, which also serves more MUs. Uniform pricing has low revenue performance due to the lack of fairness in pricing and the lack of a strategy to maximize revenue.

\begin{figure}[tp]
\captionsetup{singlelinecheck = false, justification=justified}
\centering
\includegraphics[width=3in]{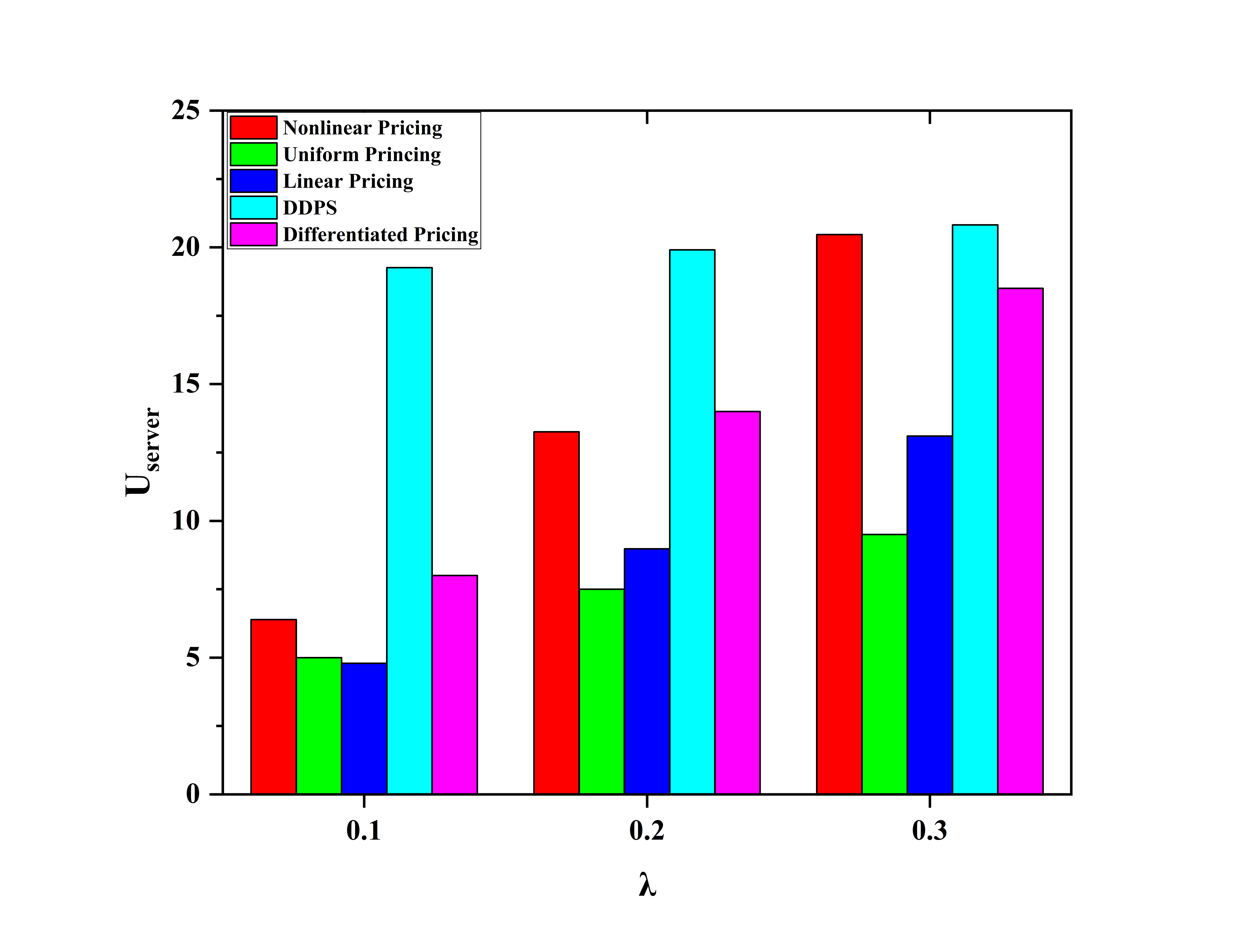}
\caption{Effect of different $\lambda$ on $U_{server}$.}
\label{fig6}
\end{figure}

\subsection{Impact of $\lambda$ on ratio of service}
Fig. \ref{fig7} illustrates the average ratio of service (RoS) for MUs. We define the RoS as a binary value, 1 represents if the MU execution delay requirements are satisfied, and 0 otherwise for each MU \cite{HHJKS}. The figure illustrates that as $\lambda$ increases, the RoS of the proposed scheme remains relatively stable, whereas the others exhibit a declining trend. This is because  Algorithm \ref{alg:alg1} of the DDPS scheme allocates the least computing resources under the premise of the maximum delay of each MU to avoid the situation that some MUs cannot be served. Therefore, compared with the strategy that MUs freely occupy server resources in other schemes (i.e., uniform pricing and differentiated pricing), DDPS serves more MUs.
Notably, the reduction is particularly large for the uniform and differential pricing schemes as $\lambda$ increases. For both pricing schemes, the rapid decrease is attributed to the fact that as the number of MUs increases, they are more likely to compete for more computing resources, which results in serving fewer MUs. However, for the linear and nonlinear pricing schemes, the price constraints prevent MUs from excessive competition for computing resources, thus their RoS is relatively stable.


In addition, the reason of the RoS value is not 1 is that it includes non-offloading case, because they cannot satisfy their requirements for offloading due to the situation such as insufficient server resources or $q_i > l_m$.
\begin{figure}[t]
\captionsetup{singlelinecheck = false, justification=justified}
\centering
\includegraphics[width=3in]{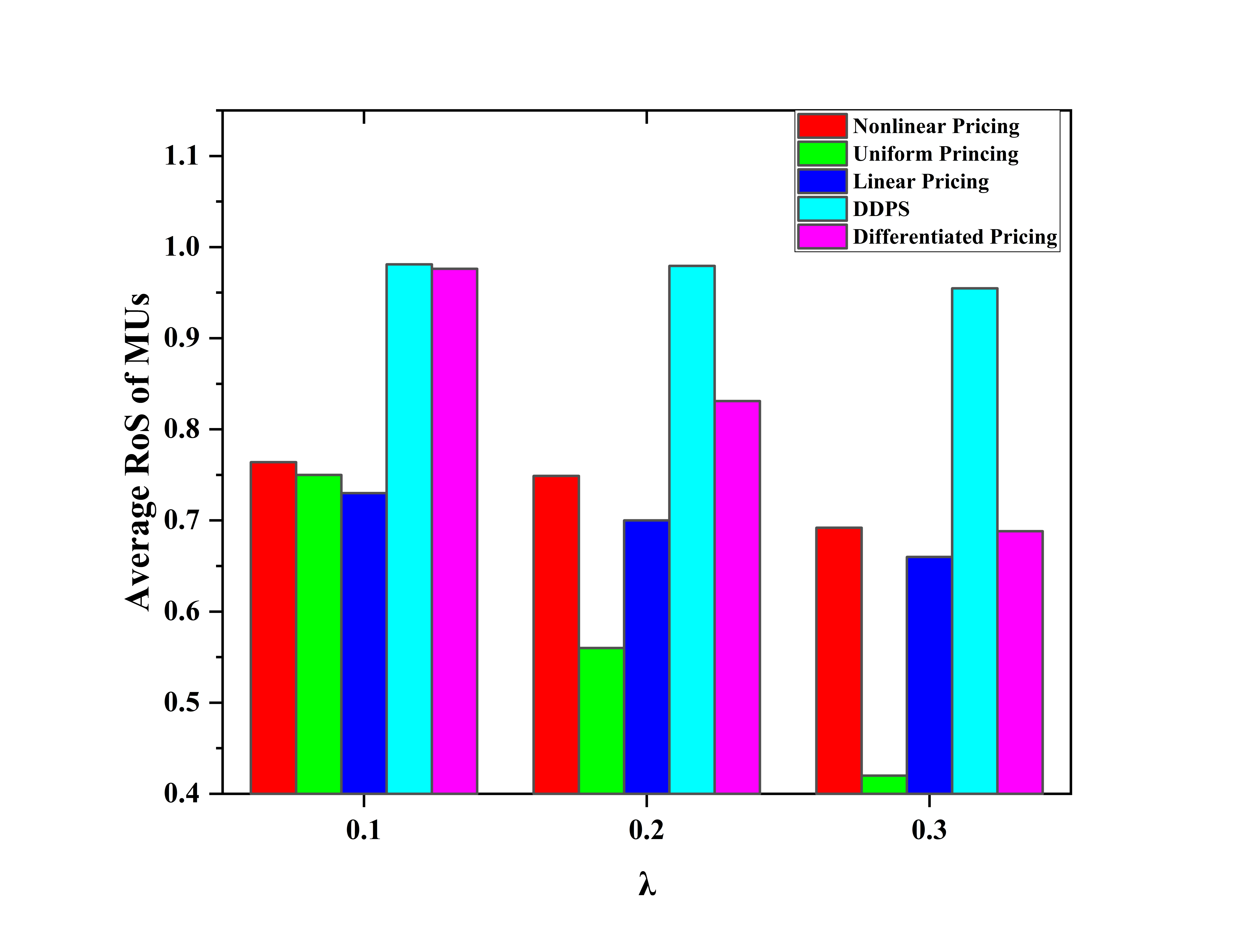}
\caption{Effect of different $\lambda$ on RoS. }
\label{fig7}
\end{figure}

\section{Conclusion and future work}



In this paper, we propose a DDPS pricing scheme as part of a task offloading scenario designed for sensor applications in remote areas. DDPS differentiates MUs based on their usage of server computing resource and the amount of offloaded data, leading to differential unit prices. Based on the proposed DDPS scheme, the optimal amount of data offloading is first determined. After that, the Stackelberg game between the MUs and the server is established considering the computing capacity requirement of the MUs. The goal of the game is to fully utilize the computing resources of the server and maximize the server's revenue. Extensive simulation results demonstrate that the DDPS scheme effectively improves the server revenue and the MUs RoS, and performs well in terms of execution delay. Therefore, through the proposed DDPS scheme, MEC service providers can reduce instances where MUs overuse server resources and serve more MUs with limited resources, ultimately improving operational efficiency. In our future work, we will investigate an efficient resource reallocation-based pricing scheme with budget constraints to enhance task processing efficiency.

\begin{appendices}
\section{} 
\label{price}
Our pricing function is $lg(\cdot)$ function, assuming that the same MU uses computing resources from $10^n$ to $10^m$, then the conversion is equivalent to the unit price from $n$ to $m$, doubled by $m/n$ times, but the computing resources are  $10^m/10^n$, that is, $10 (m-n)$ times. In other words, MU's usage of server computing resources has increased significantly, and the unit price has only changed slightly.

\section{}
\label{latency}
    The uniform pricing scheme exhibits higher latency compared to the proposed scheme.
    In the uniform pricing model, MUs are allocated computation resources corresponding to the amount of data they offload. This implies that the average execution latency is a fixed value. We assume that $\mathcal{N}$ MUs offload the same amount of data, denoted as $\mathcal{L}$, meaning that $\mathcal{N}$ Mus share an equal allocation of computation resources $\mathcal{F}$. The average computation latency is given by the following expression.
    \begin{equation}
        \overline{t_1}=\frac{h\mathcal{L}}{\mathcal{F}/\mathcal{N}}.
    \end{equation}
    
    In the FCFS model we propose, although a time slot can serve multiple MUs, we can consider $\mathcal{K}$ MUs as an entirety. This entirety forms a new FCFS queue with the subsequent entirety having $\mathcal{K}$ MUs. Therefore, the computation latency for the first entirety is $\frac{h\mathcal{L}}{\mathcal{F}/\mathcal{K}}$,\footnote{Here we assume that ${\mathcal{F}/\mathcal{K}}$ is an integer. If it is not an integer, then each MU in the last entirety will receive a larger $F_i$, resulting in lower compute latency. Therefore, we only discuss the case when it is an integer.} the latency for the second member is $\frac{2h\mathcal{L}}{\mathcal{F}/\mathcal{K}}$, and so on. The latency for the \textit{i}th member is thus $\frac{ih\mathcal{L}}{\mathcal{F}/\mathcal{K}}$. Consequently, the average computation latency is as follows.

    \begin{equation}
    \begin{aligned}
        \overline{t_2}&=(\sum_{i=1}^{\mathcal{N}/\mathcal{K}}i\frac{h\mathcal{L}}{\mathcal{F}/\mathcal{K}})/(\mathcal{N}/\mathcal{K}) \\
        &=\frac{h(\mathcal{K}+\mathcal{N})\mathcal{L}}{2F}.
    \end{aligned}
    \end{equation}
    if and only if $\mathcal{N} = \mathcal{K}$, $\overline{t_1} = \overline{t_2}$. Otherwise, $\overline{t_1} > \overline{t_2}$, i.e., the average execution delay of the proposed scheme is less than that of uniform pricing. This completes the proof.
    
\end{appendices}

\begin{IEEEbiography}[{\includegraphics[width=1in,height=1.25in,clip,keepaspectratio]{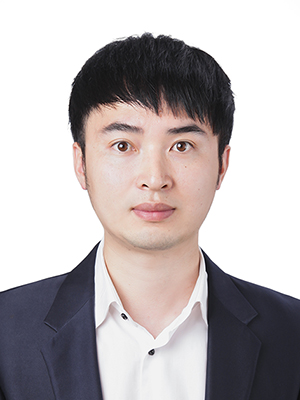}}]{Hai Xue} (Member, IEEE) 
received the B.S. degree from Konkuk University, Seoul, South Korea, in 2014, the M.S. degree from Hanyang University, Seoul, South Korea, in 2016, and the Ph.D. degree from Sungkyunkwan University, Suwon, South Korea, in 2020. He is currently an assistant professor at the School of Optical-Electrical and Computer Engineering, University of Shanghai for Science and Technology (USST), Shanghai, China. Prior to joining USST, he was a research professor with the Mobile Network and Communications Lab. at Korea University from 2020 to 2021, Seoul, South Korea. His research interests include edge computing/intelligence, SDN/NFV, machine learning, and swarm intelligence.
\end{IEEEbiography}

\vspace{-33pt}
\begin{IEEEbiography}[{\includegraphics[width=1in,height=1.25in,clip,keepaspectratio]{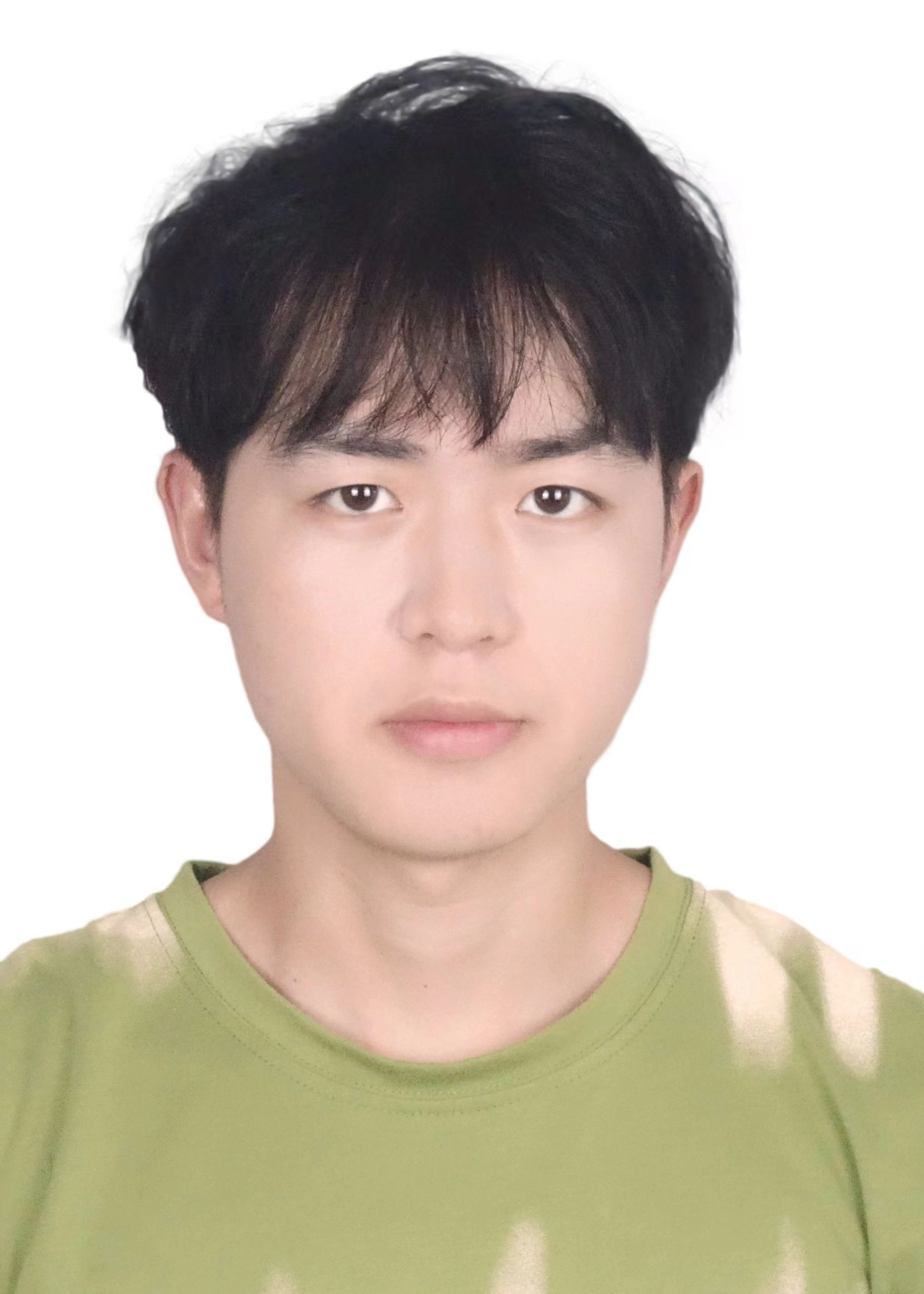}}]{Yun Xia} received the B.S. degree from Shanghai Ocean University, Shanghai, China, in 2022. He is currently working toward the master's degree in computer science and technology with the University of Shanghai for Science and Technology, Shanghai, China. His research focuses on mobile edge computing and pricing strategy.
\end{IEEEbiography}

\vspace{-33pt}
\begin{IEEEbiography}[{\includegraphics[width=1in,height=1.25in, clip,keepaspectratio]{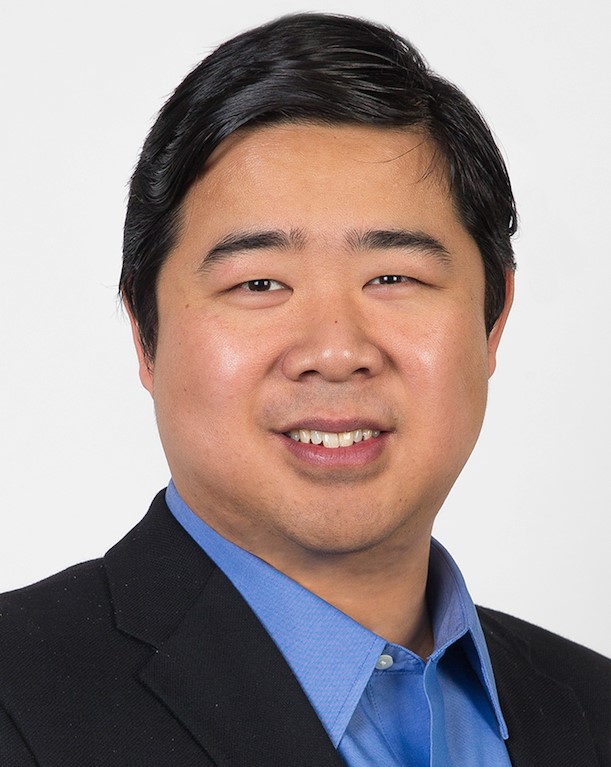}}]{Neal N. Xiong} (S'05, M'08, SM'12) is current a Professor, Computer Science Program Chair, at Department of Computer Science and Mathematics, Sul Ross State University, Alpine, TX 79830, USA. He received his both PhD degrees in Wuhan University (2007, about sensor system engineering), and Japan Advanced Institute of Science and Technology (2008, about dependable communication networks), respectively. Before he attended Sul Ross State University, he worked in Georgia State University, Northeastern State University, and Colorado Technical University (full professor about 5 years) about 15 years. His research interests include Cloud Computing, Security and Dependability, Parallel and Distributed Computing, Networks, and Optimization Theory. 

Dr. Xiong published over 250 IEEE journal papers and over 100 international conference papers. Some of his works were published in IEEE JSAC, IEEE or ACM transactions, ACM Sigcomm workshop, IEEE INFOCOM, ICDCS, and IPDPS. He has been a General Chair, Program Chair, Publicity Chair, Program Committee member and Organizing Committee member of over 100 international conferences, and as a reviewer of about 100 international journals, including IEEE JSAC, IEEE SMC (Park: A/B/C), IEEE Transactions on Communications, IEEE Transactions on Mobile Computing, IEEE Trans. on Parallel and Distributed Systems. He is serving as an Editor-in-Chief, Associate editor or Editor member for over 10 international journals (including Associate Editor for IEEE Tran. on Systems, Man \& Cybernetics: Systems, Associate Editor for IEEE Transactions on Network Science and Engineering, Associate Editor for Information Science, Editor-in-Chief for Journal of Internet Technology (JIT), and Editor-in-Chief for Journal of Parallel \& Cloud Computing (PCC)), and a guest editor for over 10 international journals, including Sensor Journal, WINET and MONET. He has received the Best Paper Award in the 10th IEEE International Conference on High Performance Computing and Communications (HPCC-08) and the Best student Paper Award in the 28th North American Fuzzy Information Processing Society Annual Conference (NAFIPS2009). 
\end{IEEEbiography}

\vspace{-33pt}
\begin{IEEEbiography}[{\includegraphics[width=1in,height=1.25in,clip,keepaspectratio]{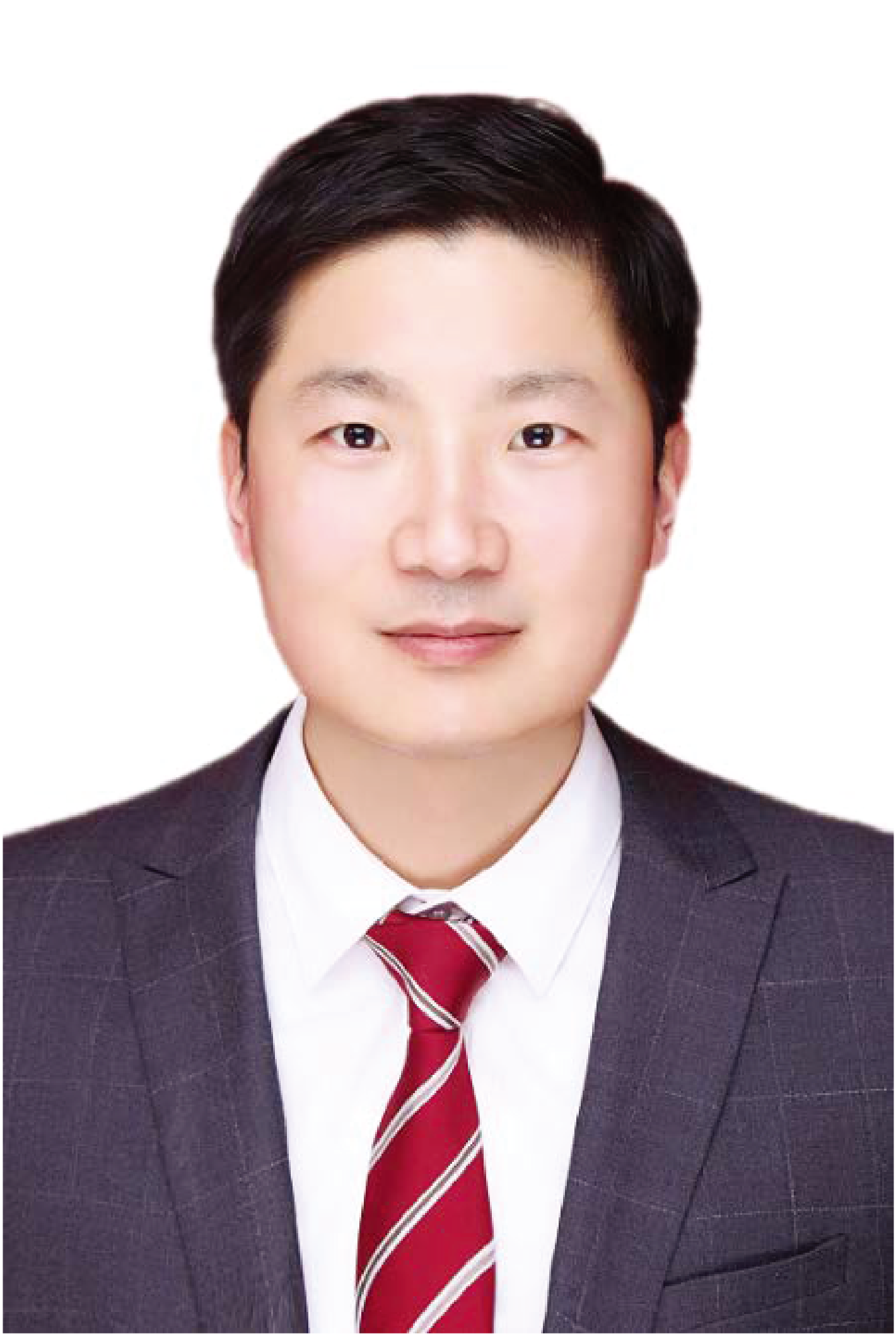}}]{Di Zhang} (S'13, M'17, SM'19)
currently is an Associate Professor at Zhengzhou University. He received his PhD degree from Waseda University. He was a Visiting Scholar of Korea University, a Senior Researcher of Seoul National University, a Visiting Student of National Chung-Hsing University. He is serving as an area editor of KSII Transactions on Internet and Information Systems, has served as the guest editor of \textsc{IEEE Wireless Communications} and \textsc{IEEE Network}. He received the First Prize Award for Science and Technology Progresses of Henan Province in 2023, the First Prize Award for Science and Technology Achievements from Henan Department of Education in 2023, and the ITU Young Author Recognition in 2019. His research interests are the wireless communications and networking, especially the short packet communications and its applications.
\end{IEEEbiography}

\vspace{-33pt}
\begin{IEEEbiography}[{\includegraphics[width=1in,height=1.25in,clip,keepaspectratio]{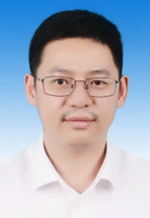}}]{Songwei Pei} (Senior Member, IEEE) received the B.S. in the School of Computer from National University of Defense and Technology, Changsha, China, in 2003, and the Ph.D. in the School of Computer Science and Technology from Fudan University, Shanghai, China, in 2009. He is currently a full professor and leading director of CAPAL laboratory within the Computer Science and Engineering Department at the University of Shanghai for Science and Technology, and he currently works as a Guest Researcher at the Institute of Computing Technology, Chinese Academy of Sciences (2018-). He had studied as Research Fellow at University of California, Irvine from 2013 to 2015, and he studied as a senior visiting scholar at Queensland University of Technology in 2017.

His research interests include intelligent computation, deep learning, heterogeneous multi-core processors, cloud computing, big data, and fault-tolerant computation, etc. He is a distinguished member of CCF, associate editor of FGCS and IEEE ACCESS, the chair of CCF YOCSEF Shanghai (2022-2023).

\end{IEEEbiography}


\begin{thebibliography}{1}
\bibliographystyle{IEEEtran}
\bibitem{HHJKS}
H. Seo, H. Oh, J. K. Choi, and S. Park, ``Differential Pricing-Based Task Offloading for Delay-Sensitive IoT Applications in Mobile Edge Computing System," {\em IEEE Internet Things J.}, vol. 9, no. 19, pp. 19116-19131,  Oct. 2022.

\bibitem{ZLXQ}
Z. Chang, L. Liu, X. Guo, and Q. Sheng, ``Dynamic Resource Allocation and Computation Offloading for IoT Fog Computing System," {\em IEEE Trans. Ind. Inf.}, vol. 17, no. 5, pp. 3348-3357, May 2021.

\bibitem{YZGST}
Y. Chen, Z. Chang, G. Min, S. Mao, and T. Hämäläinen, ``Joint Optimization of Sensing and Computation for Status Update in Mobile Edge Computing Systems," {\em IEEE Trans. Wireless Commun.}, vol. 22, no. 11, pp. 8230-8243, Nov. 2023,



\bibitem{XWSZX}
X. Chen, W. Li, S. Lu, Z. Zhou, and X. Fu, ``Efficient resource allocation for on-demand mobile-edge cloud computing,” {\em IEEE Trans. Veh. Technol.}, vol. 67, no. 9, pp. 8769–8780, Sep. 2018.


\bibitem{JSLY}
J. Yan, S. Bi, L. Duan, and Y.-J. A. Zhang, ``Pricing-driven service caching and task offloading in mobile edge computing,” {\em IEEE Trans. Wireless Commun.}, vol. 20, no. 7, pp. 4495–4512, Jul. 2021.

\bibitem{CJZ}
C. Yi, J. Cai, and Z. Su, ``A multi-user mobile computation offloading and transmission scheduling mechanism for delay-sensitive applications,” {\em IEEE Trans. Mobile Comput.}, vol. 19, no. 1, pp. 29–43, Jan. 2020.



\bibitem{MKMH}
M. Tao, K. Ota, M. Dong, and H. Yuan, ``Stackelberg Game-Based Pricing and Offloading in Mobile Edge Computing,"  {\em IEEE Wireless Commun. Lett.}, vol. 11, no. 5, pp. 883-887, May 2022.

\bibitem{YZBKK}
Y. Chen, Z. Li, B. Yang, K. Nai, and K. Li, ``A Stackelberg game approach to multiple resources allocation and pricing in mobile edge computing," {\em Future Gener Comput Syst}, vol. 108, pp. 273-287, 2020.

\bibitem{YLYDY}
Y. Li, L. Li, Y. Xia, D. Zhang, and Y. Wang, ``Multi-Leader Single-Follower Stackelberg Game Task Offloading and Resource Allocation Based on Selection Optimization in Internet of Vehicles,"  {\em IEEE Access}, vol. 11, pp. 64430-64441, 2023.

\bibitem{W}
W. Qin, C. Zhang, H. Yao, T. Mai, S. Huang, D. Guo, and R. Gao, ``Stackelberg Game-Based Offloading Strategy for Digital Twin in Internet of Vehicles," {\em in Proc. Int. Wirel. Commun. Mob. Comput. (IWCMC)},  2023, pp. 1365-1370.

\bibitem{MLPXKK}
M. Wang, L. Zhang, P. Gao, X. Yang, K. Wang, and K. Yang, ``Stackelberg-Game-Based Intelligent Offloading Incentive Mechanism for a Multi-UAV-Assisted Mobile-Edge Computing System,"  {\em IEEE Internet Things J.}, vol. 10, no. 17, pp. 15679-15689,  Sept. 2023.

\bibitem{GEES}
G. Mitsis, E. E. Tsiropoulou, and S. Papavassiliou, ``Price and Risk Awareness for Data Offloading Decision-Making in Edge Computing Systems," {\em IEEE Syst. J.}, vol. 16, no. 4, pp. 6546-6557, Dec. 2022.

\bibitem{QHTCY}
Q. Li, H. Yao, T. Mai, C. Jiang, and Y. Zhang, ``Reinforcement-Learning- and Belief-Learning-Based Double Auction Mechanism for Edge Computing Resource Allocation," {\em IEEE Internet Things J.}, vol. 7, no. 7, pp. 5976-5985, Jul. 2020.


\bibitem{THP}
T. H. Hai and P. Nguyen, ``A Pricing Model for Sharing Cloudlets in Mobile Cloud Computing," {\em in Proc. Int. Conf. Adv. Comput. Appl. (ACOMP)}, 2017, pp. 149-153.

\bibitem{ZJH}
Z. Shen, J. Zhang, and H. Tan, ``A Truthful FPTAS Auction for the Edge-Cloud Pricing Problem," {\em in Proc. 6th  Int. Conf. Big Data Comput. Commun. (BIGCOM)},  China, 2020, pp. 140-144.

\bibitem{JS}
J. S. Ng, W. Yang Bryan Lim, S. Garg, Z. Xiong, D. Niyato, M. Guizani, and C. Leung, ``Collaborative Coded Computation Offloading: An All-pay Auction Approach," {\em in Proc. IEEE Int Conf Commun}, 2021, pp. 1-6.

\bibitem{QSJCL}
Q. Wang, S. Guo, J. Liu, C. Pan, and L. Yang, ``Profit Maximization Incentive Mechanism for Resource Providers in Mobile Edge Computing," {\em IEEE Trans. Serv. Comput.}, vol. 15, no. 1, pp. 138-149,  Jan.-Feb. 2022.

\bibitem{WJYP}
W. Sun, J. Liu, Y. Yue, and P. Wang, ``Joint Resource Allocation and Incentive Design for Blockchain-Based Mobile Edge Computing," {\em IEEE Trans. Wireless Commun.}, vol. 19, no. 9, pp. 6050-6064, Sept. 2020.

\bibitem{BXYY}
B. Wu, X. Chen, Y. Chen, and Y. Lu, ``A Truthful Auction Mechanism for Resource Allocation in Mobile Edge Computing,"  {\em in Proc. IEEE 22nd Int. Symp. World Wirel., Mob. Multimed. Networks (WoWMoM)}, 2021, pp. 21-30.

\bibitem{LXXLYM}
L. Ma, X. Wang, X. Wang, L. Wang, Y. Shi, and M. Huang, ``TCDA: Truthful Combinatorial Double Auctions for Mobile Edge Computing in Industrial Internet of Things," {\em IEEE Trans Mob Comput}, vol. 21, no. 11, pp. 4125-4138,  Nov. 2022.

\bibitem{RCPYD}
R. Wang, C. Zang, P. He, Y. Cui, and D. Wu, ``Auction Pricing-Based Task Offloading Strategy for Cooperative Edge Computing," {\em in Proc. IEEE Glob. Commun. Conf. (GLOBECOM)}, 2021, pp. 01-06.

\bibitem{YWYF}
Y. Su, W. Fan, Y. Liu, and F. Wu, ``A Truthful Combinatorial Auction Mechanism Towards Mobile Edge Computing in Industrial Internet of Things," {\em IEEE Trans. on Cloud Comput.}, vol. 11, no. 2, pp. 1678-1691,  Apr.-Jun. 2023.


\bibitem{DWY}
D. Han, W. Chen, and Y. Fang, ``A Dynamic Pricing Strategy for Vehicle Assisted Mobile Edge Computing Systems," {\em IEEE Wireless Commun. Lett.}, vol. 8, no. 2, pp. 420-423, Apr. 2019.

\bibitem{ZLCYHY}
Z. Chang, C. Wang, and H. Wei, ``Flat-Rate Pricing and Truthful Offloading Mechanism in Multi-Layer Edge Computing," {\em IEEE Trans. Wirel. Commun.}, vol. 20, no. 9, pp. 6107-6121, Sept. 2021.

\bibitem{MY}
M. Liu and Y. Liu, ``Price-Based Distributed Offloading for Mobile-Edge Computing With Computation Capacity Constraints,"  {\em IEEE Wireless Commun. Lett.}, vol. 7, no. 3, pp. 420-423, Jun. 2018.



\bibitem{BRHYNA}
B. Liang, R. Fan, H. Hu, Y. Zhang, N. Zhang, and A. Anpalagan, ``Nonlinear Pricing Based Distributed Offloading in Multi-User Mobile Edge Computing," {\em IEEE Trans. Veh. Technol.}, vol. 70, no. 1, pp. 1077-1082, Jan. 2021.

\bibitem{SSMC}
S.  Kim, S. Park, M. Chen, and C. Youn, ``An Optimal Pricing Scheme for the Energy-Efficient Mobile Edge Computation Offloading With OFDMA," {\em IEEE Commun. Lett.}, vol. 22, no. 9, pp. 1922-1925, Sept. 2018.

\bibitem{LMZTQSQ}
L. Li, M. Siew, Z. Chen, and T. Q. S. Quek, ``Optimal Pricing for Job Offloading in the MEC System With Two Priority Classes,"  {\em IEEE Trans. Veh. Technol.}, vol. 70, no. 8, pp. 8080-8091, Aug. 2021.

\bibitem{QL}
Q. Yao and L. Tang, ``An Approximation Algorithm for Pricing and Request Matching in Mobile ad hoc MEC System," {\em in Proc. Comput., Commun. and IoT Applications (ComComAp)},  2019, pp. 432-437.

\bibitem{WLPXKK}
M. Wang, L. Zhang, P. Gao, X. Yang, K. Wang, and K. Yang, ``Stackelberg-Game-Based Intelligent Offloading Incentive Mechanism for a Multi-UAV-Assisted Mobile-Edge Computing System," {\em IEEE Internet Things J.}, vol. 10, no. 17, pp. 15679-15689, Sept. 2023.






\bibitem{RCPYDD}
R. Wang, C. Zang, P. He, Y. Cui, and D. Wu, ``Auction Pricing-Based Task Offloading Strategy for Cooperative Edge Computing,"  {\em in Proc. IEEE Glob. Commun. Conf. (GLOBECOM)}, Dec. 2021, pp. 01-06.

\bibitem{YMXLJ}
Y. Wang, M. Sheng, X. Wang, L. Wang, and J. Li, ``Mobile-edge computing: Partial computation offloading using dynamic voltage scaling,”
{\em IEEE Trans. Commun.}, vol. 64, no. 10, pp. 4268–4282, Oct. 2016.



\bibitem{LSYZ}
L. Dong, S. Han, Y. Gao, and Z. Tan, ``A Game-Theoretical Approach for Resource Allocation in Mobile Edge Computing," {\em in Proc. IEEE 20th Int. Conf. Commun. Technol. (ICCT)}, 2020, pp. 436-440.

\bibitem{WGS}
W. Zhang, G. Zhang, and S. Mao, ``Joint Parallel Offloading and Load
Balancing for Cooperative-MEC Systems With Delay Constraints,” {\em IEEE Trans. Veh. Technol.}, vol. 71, no. 4, pp. 4249-4263, Apr. 2022.
\bibitem{JZXRZT}
J. Chen, Z. Chang, X. Guo, R. Li, Z. Han, and T. Hämäläinen, ``Resource Allocation and Computation Offloading for Multi-Access Edge Computing With Fronthaul and Backhaul Constraints," {\em IEEE Trans. Veh. Technol.}, vol. 70, no. 8, pp. 8037-8049, Aug. 2021
\bibitem{MWXYLL}
M. Guo, W. Wang, X. Huang, Y. Chen, L. Zhang, and L. Chen, ``Lyapunov Based Partial Computation Offloading for Multiple Mobile Devices
Enabled by Harvested Energy in MEC,” {\em IEEE Internet Things J.}, vol. 9, no. 11, pp. 9025-9035, Jun. 2022.


\bibitem{MWDPLLJ}
M. DeVirgilio, W. D. Pan, L. L. Joiner, and D. Wu, ``Internet delay statistics: Measuring Internet feel using a dichotomous Hurst parameter,” {\em in Proc. IEEE SoutheastCon}, 2012, pp. 1-6.


\bibitem{NAKTM}
N. Powers, A. Alling, K. Osolinsky, T. Soyata, M. Zhu, H. Wang,
H. Ba, W. Heinzelman, J. Shi, and M. Kwon, ``The Cloudlet Accelerator: Bringing Mobile-Cloud Face Recognition into Real-Time,” {\em in Proc. IEEE Globecom Workshops (GC Wkshps)},  2015, pp. 1-7.




 
\end{thebibliography}
\end{document}